\documentclass[12pt]{article}

\usepackage{hyperref}

\usepackage{amsmath} 
\usepackage{amssymb}
\usepackage{amsfonts}
\usepackage{amsthm}
\usepackage{graphicx}

\newtheorem{theorem}{Theorem}
\newtheorem{lemma}{Lemma}

\theoremstyle{definition}

\begin{document}
	\begin{center}
		\Large
	\textbf{Integral representation of finite temperature non-Markovian evolution of some RWA systems}\footnote{This work is supported by the Russian Science Foundation under grant 19-11-00320.}	
	
		\large 
		\textbf{A.E. Teretenkov}\footnote{Department of Mathematical Methods for Quantum Technologies, Steklov Mathematical Institute of Russian Academy of Sciences,
			ul. Gubkina 8, Moscow 119991, Russia\\ E-mail:\href{mailto:taemsu@mail.ru}{taemsu@mail.ru}}
		\end{center}
		
			\footnotesize
			We introduce the Friedrichs model at finite temperature which is one- and zero-particle restriction of spin-boson in the rotating wave approximation and obtain the population of the excited state for this model. We also consider the oscillator interacting with bosonic thermal bath in the rotating wave approximation and obtain dynamics of mean excitation number for this oscillator. Both solutions are expressed in terms of integrals of zero-temperature solutions with correspondent correlation functions. 
			\normalsize

	\section{Introduction}
	
	In recent years there has been growth of interest in the problems of non-Markovian evolution of open quantum systems. Both pure mathematical questions \cite{Kossakowski07, Breuer07, Breuer09, Gullo14,  Rivas14} and physical applications \cite{Mohseni08, Plenio08, Kolli12, Singh12, Plenio13, Bylicka16} has led to this growth. Many theoretical works use the spin-boson model in the rotating wave approximation (RWA) at zero temperature as a paradigmatic exactly solvable model to test different approaches or to formulate hypotheses which could be generalized to other models. For example, the circle of convergence of time-convolutionless perturbation theory was obtained for this model. Both BLP \cite{Breuer09} and RHP \cite{Rivas14} measures of non-Markovianity become non-zero at the same coupling strength \cite{Haikka11}. Simulation based on tensor networks was tested for this model \cite{Luchnikov19}. A generalization of this model was used to analyze the applicability of the master equations in the second order (Born) approximation \cite{Teretenkov19}. For the special case of the spectral density, which is a combination of the Lorentz peaks, the non-Markovian evolution for this model could be dilated to  Gorini--Kossakowski--Sudarshan--Lindblad evolution with a finite number of degrees of freedom \cite{Teretenkov19m}. This dilation leads to the so-called pseudomode approach \cite{Imamoglu94, Garraway96, Garraway97, Garraway97a, Dalton01, Garraway06, Schonleber15, Teretenkov19, Teretenkov19m} which has been generalized to models without RWA \cite{Tamascelli18,Mascherpa19} and non-zero temperature reservoirs \cite{Tamascelli19} only recently. Also the generalization to other spectral densities was discussed in \cite{Pleasance20}.
	
	For the zero temperature the solution of the RWA spin-boson is reduced to one-particle subspace evolution which by the Friedrichs model \cite{Friedrichs48}. This is not the case for the finite-temperature RWA spin-boson. Since our goal is to  take somehow into account the effects of reservoir finite temperature, but still have an exactly solvable model which is analogous to the zero-temperature one, we have introduced a compromise model in  Sec.~\ref{sec:Fried} which we call the Friedrichs model at finite temperature. Namely, we consider the initial condition which is a one- and zero- particle restriction of the actual Gibbs thermal state. The main result for this model is presented in Th.~\ref{th:exSt}, where we obtain population of excited state. It is expressed in terms of integrals of zero-temperature solution with the restricted thermal reservoir correlation function.
	
	Similar to the zero-temperature case \cite{Teretenkov19m} the Friedrichs model at finite temperature is one- and zero- particle restriction not only of the RWA spin-boson, but also of other RWA models such as the RWA oscillator which also interacts with the boson reservoir. But the RWA oscillator also has exact solution for thermal reservoir initial condition without restriction to the one- and zero- particle subspaces. We consider this model in Sec.~\ref{sec:osc}. The main result for the RWA oscillator is presented in Th.~\ref{th:osc} and it is similar to the Friedrichs model. The only difference from Th.~\ref{th:exSt} consists in change of the restricted thermal reservoir correlation function to the one without restriction \cite{Teretenkov19OnePart}. This could be used to hypothesize the general behavior of finite temperature RWA models.
	
	We could consider arbitrary couplings  and time scales for both models and, hence, describe their non-Markovian behavior.  But let us stress that we do not discuss the applicability of RWA itself which was also discussed in literature \cite{Tang13, Fleming10, Trubilko2020} and could limit the range of applicability of these models to the real physical systems for arbitrary couplings  and time scales.  We sum up our results and make remarks for further discussion in Conclusions.

	\section{Friedrichs model at finite temperature}
	\label{sec:Fried}
	
	The Friedrichs model usually occurs as zero- and one-particle restriction of different models in which the number of particles is conserved \cite{Garraway97,Teretenkov19}. For quantum optical models it means the rotating wave approximation (RWA). For zero temperature reservoir it could be used to obtain exact dynamics of spin-boson and other models \cite{Garraway97, Garraway97a, Teretenkov19m}. This approach could be generalized for multi-level systems interacting with multiple reservoirs \cite{Teretenkov19}. For finite temperature not only zero-particle and one-particle subspaces have to be taken into account, but all other N-particle subspaces. This is the result of the fact that the initial state is supported on all the N-particle subspaces if the reservoir is initially in the Gibbs state at finite temperature. But if the temperature is low (optical case), then it is natural to consider the situation, when only one-particle subspace is excited. So as the initial condition for the Friedrichs model we take the one-particle restriction of the density matrix for finite temperature reservoir.
	
	Now let us formulate our model mathematically. We consider the Hilbert space $ \mathbb{C} \oplus \mathbb{C} \oplus \mathcal{L}_2(\mathbb{R}) $. Quantum dynamics is described by the Liouville-von Neumann equation
	\begin{equation}\label{eq:LioNeum}
		\frac{d}{dt} \rho(t) =- i[ \hat{H}_F, \rho(t)], 
	\end{equation}
	with Hamiltonian $ \hat{H}_F = 0 \oplus H_F $, where $ H_F $ is the operator in $  \mathbb{C} \oplus \mathcal{L}_2(\mathbb{R}) $, defined by the formula
	\begin{equation*}
		H_F = \int \omega_k |k \rangle  \langle k|  d k + \Omega |1 \rangle \langle 1| +  \int \left(  g_{ k}^*  | k \rangle \langle 1 | +   g_{k}  | 1 \rangle \langle k | \right) d k.
	\end{equation*}
	Then take as the initial condition
	\begin{equation*}
		\rho(0) =  \frac{1}{Z}\left((1-p) |0 \rangle \langle 0| + p |1 \rangle \langle 1| + (1-p)\int e^{- \beta \omega_k} |k \rangle  \langle k|  d k \right)
	\end{equation*}
	with normalization constant $	Z = e^{- \int dk \ln(1 - e^{- \beta \omega_k})} $, $ \beta > 0 $, $ \omega_k >0 $, $ p \in [0,1] $. This initial condition is a one- and zero-particle restriction of tensor product of the diagonal state and the thermal state for the spin-boson model. This leads to the fact that the initial trace of the density matrix could generally be smaller than 1.
	
	To solve this equation we solve the Schroedinger equation first and then obtain the solution of Liouville-von Neumann equation \eqref{eq:LioNeum} by averaging the pure state solution according to the initial state.  Let $ | \psi_1 (t) \rangle $ be a solution of the Schroedinger equation
	\begin{equation}\label{eq:SchroFr}
		\frac{d}{dt}| \psi (t) \rangle = - i H_F | \psi (t) \rangle
	\end{equation}
	with the initial condition $ | \psi (0) \rangle = | 1 \rangle $. Analogously, $ | \psi_k (t) \rangle $ is the solution of this Schroedinger equation with the initial condition $ | \psi (0) \rangle = | k \rangle $. Then
	\begin{equation*}
		\rho(t) = \frac{1}{Z} \left((1-p) |0 \rangle \langle 0| + p | \psi_1 (t) \rangle \langle \psi_1 (t)| +  (1-p) \int dk e^{- \beta \omega_k} | \psi_k (t) \rangle \langle \psi_k (t) |\right)
	\end{equation*}
	Thus, the coherences are not excited if they are absent in the initial condition
	\begin{equation}\label{eq:popOfExState}
		\rho_{11}(t) = \frac{1}{Z} \left( p |\psi_{1,1} (t)|^2 +  (1-p) \int dk e^{- \beta \omega_k} |\psi_{1,k}(t)|^2\right),
	\end{equation}
	where $ \psi_{1,1} = \langle  1 | \psi_1 (t) \rangle  $,  $ \psi_{1,k} = \langle  1 | \psi_k (t) \rangle  $.
	
	\begin{lemma}\label{lem:intDiff}
		Let the integrals
		\begin{equation}\label{eq:defGandF}
			G(t)  =  \int dk |g_k|^2 e^{ -i  \omega_k t} , \qquad f(t) =  \int dk  g_{k} e^{-i  \omega_k t} \psi_k(0)
		\end{equation}
		converge for all $ t \in \mathbb{R}_+ $ and define the continuous functions $ G(t) $ and $ f(t) $, then
		$ \psi_1(t) \equiv \langle 1 | \psi (t) \rangle $ is the (unique) solution of the equation
		\begin{equation}\label{eq:integroDiff}
			\frac{d}{dt} \psi_1(t) =-i \Omega \psi_1(t) - i f(t) -\int_0^t d\tau G(t - \tau)\psi_1 (\tau)
		\end{equation}
		with the initial condition $  \psi_1(0) = \langle 1 | \psi (0) \rangle  $.
	\end{lemma}
	
	\begin{proof}
		Substituting  $ | \psi (t) \rangle = \psi_1(t) | 1 \rangle +  \int dk \psi_k(t) |k \rangle $ into \eqref{eq:SchroFr} one obtains the system
		\begin{equation*}
			\begin{cases}
				\frac{d}{dt} \psi_1(t) =-i \Omega \psi_1(t) - i \int dk  g_{k} \psi_k(t)\\
				\frac{d}{dt} \psi_k(t) = - i  \omega_k \psi_k(t) - i  g_{k}^* \psi_1(t)
			\end{cases}
		\end{equation*}
		Solving the second equation as a linear differential equation for $ \psi_k(t)  $ considering $ - i  g_{k}^* \psi_1(t) $ as non-homogeneous term one has
		\begin{equation*}
			\psi_k(t) =  e^{- i  \omega_k t} \psi_k(0) - i \int ds \; e^{- i  \omega_k (t-s)}  g_{k}^* \psi_1(s).
		\end{equation*}
		Substituting this equation into the first one of the system we obtain 
		\begin{equation*}
			\frac{d}{dt} \psi_1(t) =-i \Omega \psi_1(t) - i \int dk  g_{k} e^{- i  \omega_k t} \psi_k(0) -  \int ds \; \int dk \; e^{- i  \omega_k (t-s)}  |g_{k}|^2 \psi_1(s) 
		\end{equation*}
		Taking into account definition \eqref{eq:defGandF} of the functions $ G(t) $ and $ f(t) $ we obtain \eqref{eq:integroDiff}. If these functions are continuous  \cite[Sec. 2.1]{Burton05}, then the Cauchy problem \eqref{eq:integroDiff} with initial condition $  \psi_1(0) = 1  $ has the unique solution.
	\end{proof}
	
	Let us note that if $ G(t) $ is a generalized function, the formal analog of \eqref{eq:integroDiff} has no immediate meaning and the additional regularization  by the counter-term in the system energy \cite{Teretenkov19OnePart} is needed.
	
	Recall that if $ |G(t)| $ could be bounded by the function  $ M e^{\alpha t} $ for some $ M>0 $ and $ \alpha >0 $, then the function $ G(t) $ is said to be of exponential order \cite[Sec. 2.3]{Burton05}.
	
	\begin{lemma}
		\label{lem:solOfIntDif}
		Let $ G(t) $ be of exponential order and $ x(t) $ be the solution of the equation
		\begin{equation}\label{eq:xSol}
			\frac{d}{dt} x(t) =-i \Omega x(t)  -\int_0^t d\tau G(t - \tau)x(\tau)
		\end{equation}
		with the initial condition $ x(0) = 1 $. Then
		\begin{equation}\label{eq:psiSol}
			\psi_{1,0}(t) =x(t), \qquad \psi_{1,k}(t) =- i  \int_0^t x(t-s) g_k e^{-i \omega_k s} ds.
		\end{equation}
	\end{lemma}
	
	\begin{proof}
		By Th.~2.3.1 from \cite[Sec. 2.3]{Burton05} the solution of equation \eqref{eq:integroDiff} could be represented in terms of Eq.~\eqref{eq:xSol} as
		\begin{equation*}
			\psi_1(t) =x(t) \psi_1(0) - i  \int_0^t x(t-s) f(s) ds
		\end{equation*}
		if $ G(t)$ and  $ f(t) $ are of exponential order. Taking $ f(t) =0 $ and $ f(t) = g_k e^{-i \omega_k t}$ which are of exponential order  one obtains \eqref{eq:psiSol}. 
	\end{proof}
	
	\begin{theorem}\label{th:exSt}
		Let $ G(t) $ defined by \eqref{eq:defGandF} be continuous function of exponential order and 
		\begin{equation*}
			G_{\beta}^{1}(t) \equiv \int dk |g_k|^2 e^{ -i  \omega_k t} e^{- \beta \omega_k}
		\end{equation*}
		converges and define the continuous function, then the population of the exited state defined by \eqref{eq:popOfExState} takes the form
		\begin{equation}\label{eq:exSt}
			\rho_{11}(t) = \frac{1}{Z} \left( p |x (t)|^2 +  2 (1-p) \, \mathrm{Re} \int_0^t  d\tau \int_0^{\tau}   ds \; G_{\beta}^{1}(\tau - s)   x^*(s) x(\tau) \right).
		\end{equation}
	\end{theorem}
	
	\begin{proof}
		By substituting \eqref{eq:psiSol} into \eqref{eq:popOfExState} we obtain
		\begin{align*}
			\rho_{11}(t) &= \frac{1}{Z} \left( p |x (t)|^2 +  (1-p) \int dk e^{- \beta \omega_k} \int_0^t x(t-\tau) g_k e^{-i \omega_k \tau} d\tau \int_0^t ds x^*(t-s) g_k^* e^{-i \omega_k (-s)} ds \right) \\
			&= \frac{1}{Z} \left( p |x (t)|^2 +  (1-p)  \int_0^t d\tau \int_0^t ds x(t-\tau) G(\tau - s)   x^*(t-s) ds \right)
		\end{align*}
		Taking into account $ 	G_{\beta}^{1}(-t) = (G_{\beta}^{1}(t))^* $ we obtain \eqref{eq:exSt}.
	\end{proof}
	
	It is natural to call $ G_{\beta}^{1}(t) $ the restricted thermal reservoir correlation function. Let us also note that $ G_{\beta}^{1}(t) $ could be considered as analytic continuation  $ G(t- i \beta), \beta >0 $ of $ G(t) $ to lower half-plane.
	
	In physics usually the spectral density  $ \mathcal{J}(\omega) $ \cite[Subec.~3.6.2.1]{Breuer02}, \cite{Garraway97} is defined by 
	\begin{equation*}
		G(t) = \int_{0}^{+\infty} \frac{d \omega}{2 \pi} e^{- i \omega t}  \mathcal{J}(\omega).
	\end{equation*}
	Let us stress that in contrast to  \cite{Garraway97} and \cite{Teretenkov19} the lower limit of integral here is $ 0 $ rather than $ -\infty $. This is because the positivity of function $ \omega_k $ is important here due to the fact that we need a thermal state to be normalizable, which is usually neglected for zero-temperature RWA systems. Note that we do not use odd continuation of the spectral density which is widespread \cite{Trushechkin19, Trushechkin19a}. 
	
	\begin{figure}[h]
		\includegraphics[width=0.85\textwidth]{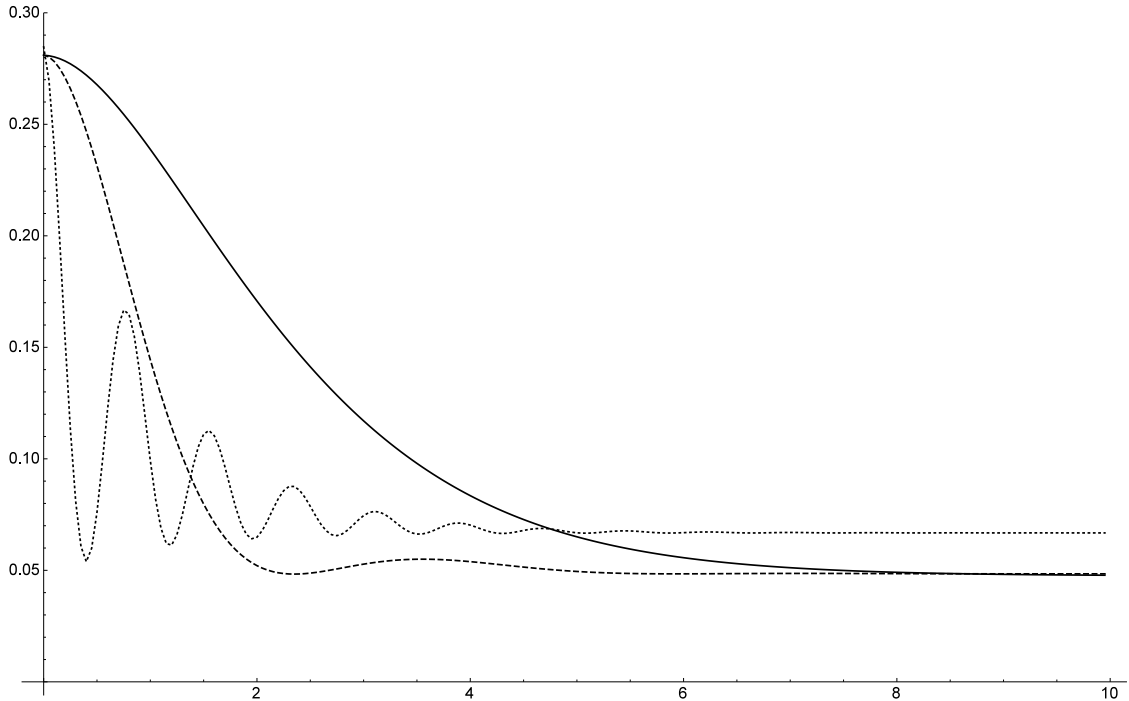}
		\caption{Dependence of population $ \rho_{11} $ on time $ t $ (The abscissa axis represents $ \gamma t $.) for different coupling strength. Solid line for $ g/\gamma =  0.5 $, dashed line for $ g/\gamma =  1 $ and dotted line for $ g/\gamma =  4 $. }\label{fig:1}
	\end{figure}
	
	Let us consider the resonance case  
	\begin{equation*}
		\mathcal{J}(\omega) = \frac{\gamma g^2}{ \left(\frac{\gamma}{2}\right)^2 + (\omega - \Omega)^2}, \quad g>0, \gamma > 0,
	\end{equation*}
	i.e. the  Lorentz spectral density centered at frequency $ \Omega $ of system free evolution. It is usually assumed that $ \Omega $ is high so one could integrate from $ - \infty $ to obtain $ G(t) $, which to leads to exponential form of $ G(t) $ and possibility to solve \eqref{eq:xSol} analytically. But now we need also $ G_{\beta}^{1}(t) $ to converge and the solution could not be represented analytically (in elementary functions) in this case, so we have implemented the solution of \eqref{eq:xSol} and formula \eqref{eq:exSt} numerically. The result for $ \Omega/\gamma = 5 $, $ p=0.3 $, $ \beta \gamma = 0.5 $, $ \omega_k =100 |k| $ (it is needed only to calculate $ Z $) and several values of the coupling constant, which is presented in Fig.~\ref{fig:1}. 
	
	As expected for the finite temperature the population does not tend to zero, but has a finite limit at long times. It is interesting that for strong coupling this limit depends not only on the reservoir temperature, but also on the coupling constant. As for zero-temperature \cite[Sec.~10.1]{Breuer02} at weak coupling the tendency to the stationary state is monotonic and at strong coupling it becomes oscillatory.

	\section{RWA oscillator at finite temperature}
	\label{sec:osc}
	
	Due to the fact that an oscillator in the bosonic bath is exactly solvable without RWA (it is the famous Caldeira-Leggett model), the RWA oscillator is not so widely discussed as RWA spin-boson, but there are it is also discussed in literature \cite{Agarwal71}. 
	
	Now let us introduce mathematical description of the model. We consider the Hilbert space $ \ell_2 \otimes \mathcal{L}_2(\mathbb{R}) $. One could define creation and annihilation operators in this space satisfying canonical commutation realations: $ [a, a^{\dagger}] =1 $, $ [b_k, b_{k'}^{\dagger}] = \delta(k-k') $, $ [b_k, b_{k'}] = 0  $, $ [a, b_k] =0 $ and $ [a, b_k^{\dagger}] =0 $. The Hamiltonian of RWA oscillator
	\begin{equation*}
		\hat{H} =  \Omega a^{\dagger} a +  \int \omega_k b_k^{\dagger} b_k dk +   \int( g_k^* a b_k^{\dagger} + g_k a^{\dagger} b_k)dk.
	\end{equation*}
	
	\begin{lemma}\label{lem:Heis}
		The operator $ a(t) =  e^{i \hat{H} t} a e^{-i \hat{H} t} $, i.e. the Heisenberg evolution of operator $ a $, satisfies the equation
		\begin{equation}\label{eq:HeisIntDiff}
			\frac{d}{dt} a(t) = - i \Omega a(t) - i  b(t) - \int_0^{t} G(t-s) a(s) ds,
		\end{equation}
		where
		\begin{equation*}
			b(t) \equiv \int e^{- i \omega_k t} g_k b_k(0) dk .
		\end{equation*}
	\end{lemma}
	
	\begin{proof}
		Evaluating the commutators $ [\hat{H}, a] $ and $  [\hat{H}, b_k] $ we obtain the Heisenberg equations
		\begin{equation*}
			\begin{cases}
				\frac{d}{dt} a(t) = - i \Omega a(t) - i \int g_k b_k(t) dk,\\
				\frac{d}{dt} b_k(t) = - i \omega_k b_k(t) - i g_k^* a(t).
			\end{cases}
		\end{equation*}
		The end of the proof is analogous to that of lemma~\ref{lem:intDiff} except for the fact that we do not discuss the existence and uniqueness of the solution.
	\end{proof}
	
	\begin{lemma}
		Let the reservoir be initially in the thermal state $ \rho_{\beta} $ at inverse temperature $ \beta $, i.e. in the Gaussian state defined by the following moments
		\begin{equation}\label{eq:thermalState}
			\langle  b_k(0) \rangle   = 0, \qquad\langle  b_k(0) b_{k'}(0) \rangle   = 0, \qquad \langle  b_k^{\dagger}(0) b_{k'}(0) \rangle   = \frac{1}{e^{\beta \omega_k} - 1} \delta (k - k'),
		\end{equation}
		where $ \langle \cdot \rangle $ is the averaging with respect to $ \rho_{\beta} $,then
		\begin{equation*}
			\langle  b^{\dagger}(t) b(s) \rangle =
			\begin{cases}
				G_{\beta}^{\infty}(t-s), & t \geqslant s,\\
				(G_{\beta}^{\infty}(s-t))^*, & s>t,
			\end{cases}
		\end{equation*}
		where
		\begin{equation}\label{eq:Gbeta}
			G_{\beta}^{\infty}(t) =   \int dk |g_k|^2 e^{ -i  \omega_k t}\frac{1}{e^{\beta \omega_k} - 1} .
		\end{equation}
	\end{lemma}
	
	\begin{lemma}	
		Let  $ G(t) $, $ G_{\beta}^{\infty}(t) $ be of exponential order and $ x(t) $ be a solution of integro-differential equation \eqref{eq:xSol} with the initial equation $ x(0) = 1 $, then the solution of Eq.~\eqref{eq:HeisIntDiff} has the form
		\begin{equation}\label{eq:aHeisSol}
			a(t) =x(t) a(0) - i \int_0^t x(t-\tau)   b(\tau) d \tau,
		\end{equation}
		which we understand in the sense of averaging over the state $ \rho_{a} \otimes \rho_{\beta} $, where $ \rho_a $ is an arbitrary density matrix in $ \ell_2 $ and $  \rho_{\beta} $ is defined by \eqref{eq:thermalState}. 
	\end{lemma}
	
	\begin{proof}
		The proof of this lemma is analogous to the one of lemma~\ref{lem:solOfIntDif}, but now we need to specify in what sense the operator-valued function $ b(t) $ is of exponential order. For our purposes it is enough to specify that  $ b(t) $ is of exponential order if arbitrary finite order moments are  of exponential order. Then \eqref{eq:aHeisSol} should also be understood in the sense of averaging over the state $ \rho_{a} \otimes \rho_{\beta} $ as the quantum stochastic equations raised in the stochastic limit are usually understood \cite{Accardi2002}. 
	\end{proof}
	
	We call $ G_{\beta}^{\infty}(t) $ the thermal reservoir correlation function. Let us note, that it could be represented as
	\begin{equation*}
		G_{\beta}^{\infty}(t)  = \sum_{n=1}^{\infty} G_{n\beta}^{1}(t) = \sum_{n=1}^{\infty} G(t - i n \beta)
	\end{equation*}
	in the case, when this series converges. This representation shows that the restricted thermal correlation function is the first term in this series.

	\begin{theorem}
		\label{th:osc}
		The first and second moments of the oscillator creation and annihilation operators evolving by Heisenberg equations specified in lemma \ref{lem:Heis} with respect to the state $ \rho_{a} \otimes \rho_{\beta} $, where $ \rho_a $ is an arbitrary density matrix in $ \ell_2 $ and $  \rho_{\beta} $ is defined by \eqref{eq:thermalState}, could be calculated by the following formula
		\begin{align}
			\label{eq:aEvol}
			\langle a(t) \rangle &=x(t) \langle a(0) \rangle\\
			\label{eq:adaEvol}
			\langle a^{\dagger}(t) a(t) \rangle  &=|x(t)|^2 \langle a^{\dagger}(0) a(0) \rangle + 2 \mathrm{Re} \int_0^t  \int_0^{\tau} x^*(s)   x(\tau) G_{\beta}^{\infty}(\tau - s) d s d\tau\\
			\label{eq:aaEvol}
			\langle ( a(t) )^2\rangle  &= (x(t))^2 \langle ( a(0) )^2 \rangle 
		\end{align}
		where $ \langle \cdot \rangle $ is the averaging with respect to $ \rho_{a} \otimes \rho_{\beta} $, $ x(t) $ is defined by \eqref{eq:xSol}, where $ G(t) $ defined by \eqref{eq:defGandF} is continuous and of exponential order, $ G_{\beta}^{\infty}(t) $  defined by \eqref{eq:Gbeta} is continuous.
	\end{theorem}
	
	\begin{proof}
		Let us calculate $ \langle a^{\dagger}(t) a(t) \rangle  $ by averaging 
		\begin{equation*}
			\langle a^{\dagger}(t) a(t) \rangle  =|x(t)|^2 \langle a^{\dagger}(0) a(0) \rangle + \int_0^t  \int_0^t x^*(t-\tau)   x(t-s) \langle b^{\dagger}(\tau) b(s) \rangle d\tau d s
		\end{equation*}
		Let us transform the second summand in the following way
		\begin{align*}
			\int_0^t  \int_0^{\tau} x^*(t-\tau)   x(t-s) \langle b^{\dagger}(\tau) b(s) \rangle d\tau d s + \int_0^t  \int_{\tau}^{t} x^*(t-\tau)   x(t-s) \langle b^{\dagger}(\tau) b(s) \rangle d\tau d s=\\
			= 2 \mathrm{Re}\; \int_0^t  \int_0^{\tau} x^*(t-\tau)   x(t-s) G_{\beta}^{\infty}(\tau - s) d\tau d s= 2 \mathrm{Re}\;  \int_0^t  \int_0^{\tau} x^*(s)   x(\tau) G_{\beta}^{\infty} (\tau - s) d s d\tau 
		\end{align*}
		Thus, we have obtained \eqref{eq:adaEvol}. Formula \eqref{eq:aaEvol} could be obtained analogously. Formula \eqref{eq:aEvol} is the direct result of averaging of \eqref{eq:aHeisSol}.
	\end{proof}
	
	On the other hand, one could restrict RWA oscillator to zero- and one-particle subspaces and the average $ \langle a^{\dagger}(t) a(t) \rangle $ will be defined by formula \eqref{eq:exSt} in this case, i.e. only $  G_{\beta}^{\infty} (t) $ is changed to $ \frac{1}{Z} (1-p) G_{\beta}^{1} (t) $ but both solutions are expressed in terms of zero-temperature ones in the similar way. Hence, we hypothesize that this integral relation between zero-particle solution of RWA models and finite-temperature could be obtained in a more general case with appropriate choice of the analog of the thermal correlation function.
	
	\section{Conclusions}
	
	We have discussed two models which describe the finite temperature dynamics of RWA systems. These models are exactly solvable in the same sense as their finite temperature analogs are, i.e. their solution expressed in terms of finite dimensional (actually, scalar) integro-differential equations. Moreover, these integro-differential equations are the same as the ones at zero temperature, but we need additionally integrate them with thermal correlation functions. The Friedrichs model at finite temperature could be easily generalized to the multilevel system interacting with several reservoirs as it was done in \cite{Teretenkov19} for zero-temperature. For both systems the case of Ohmic spectral density could be also described by the methods similar to \cite{Teretenkov19exact}. But the main direction for further development is the generalization of these results for wider classes of RWA models. Also the possibility of dilation to the Gorini--Kossakowski--Sudarshan--Lindblad equation with a finite number of degrees of freedom for certain spectral densities is interesting due to the modern discussion of the pseudomode approach \cite{Tamascelli19, Pleasance20}.

	\section{Acknowledgments}
	The author thanks  A.\,S.~Trushechkin for the fruitful discussion of the problems considered in the work.


\begin{thebibliography}{10}
	
\bibitem{Kossakowski07}

A.~Kossakowski and R.~Rebolledo, \textquotedblleft On non-Markovian time evolution in open quantum systems,\textquotedblright\;Open Sys. and Inform. Dynamics \textbf{14} (3), 265--274 (2007).

\bibitem{Breuer07}

H.~P. Breuer, \textquotedblleft Non-Markovian generalization of the Lindblad theory of open quantum systems,\textquotedblright\;Phys. Rev. A \textbf{75} (2), 022103 (2007).

\bibitem{Breuer09}

H.~P. Breuer, E.~M. Laine, and J.~Piilo, \textquotedblleft Measure for the degree of non-Markovian behavior of quantum processes in open systems,\textquotedblright\;Phys. Rev. Lett. \textbf{103} (21), 210401 (2009).

\bibitem{Gullo14}

N.~Lo Gullo, I.~Sinayskiy, T.~Busch, and F.~Petruccione, \textquotedblleft Non-Markovianity criteria for open system dynamics,\textquotedblright\;arXiv:1401.1126 (2014).

\bibitem{Rivas14}

A.~Rivas, S.~F. Huelga, M.~B. Plenio, \textquotedblleft Quantum non-Markovianity: characterization, quantification and detection,\textquotedblright\;Rep. on Progr. in Phys. \textbf{77} (9), 094001 (2014).

\bibitem{Mohseni08}

M.~Mohseni, P.~Rebentrost, S.~Lloyd, and A.~Aspuru-Guzik, \textquotedblleft Environment-assisted quantum walks in photosynthetic energy transfer,\textquotedblright\;J. of Chem. Phys. \textbf{129} (17), 11B603 (2008).

\bibitem{Plenio08}

M.~B. Plenio and S.~F. Huelga, \textquotedblleft  Dephasing-assisted transport: quantum networks and biomolecules,\textquotedblright\;New J. of Phys. \textbf{10} (11), 113019 (2008).

\bibitem{Kolli12}
				
A.~Kolli, E.~J. O'Reilly, G.~D. Scholes, and A.~Olaya-Castro, \textquotedblleft The fundamental role of quantized vibrations in coherent light harvesting by cryptophyte algae,\textquotedblright\;J. of Chem. Phys. \textbf{137} (17), 174109 (2012).

\bibitem{Singh12}

N.~Singh and P.~Brumer, \textquotedblleft Efficient computational approach to the non-Markovian second order quantum master equation: electronic energy transfer in model photosynthetic systems,\textquotedblright\;Mol. Phys. \textbf{110} (15--16), 1815--1828 (2012).

\bibitem{Plenio13}

M.~B. Plenio, J.~Almeida, and S.~F. Huelga, \textquotedblleft Origin of long-lived oscillations in 2D-spectra of a quantum vibronic model: electronic versus vibrational coherence,\textquotedblright\;J. of Chem. Phys. \textbf{139} (23), 12B614\_1 (2013).

\bibitem{Bylicka16}

B.~Bylicka, M.~Tukiainen, D.~Chruscinski, J.~Piilo, and S.~Maniscalco, \textquotedblleft Thermodynamic power of non-Markovianity,\textquotedblright\;Sci. Rep. \textbf{6} (1), 1--7 (2016).

\bibitem{Haikka11}

P.~Haikka, J.~D. Cresser, and S.~Maniscalco, \textquotedblleft Comparing different non-Markovianity measures in a driven qubit system,\textquotedblright\;Phys. Rev. A \textbf{83}(1), 012112 (2011).

\bibitem{Luchnikov19}

I.~A. Luchnikov, S.~V. Vintskevich, H.~Ouerdane, and S.~N. Filippov, \textquotedblleft Simulation complexity of open quantum dynamics: Connection with tensor networks,\textquotedblright\;Phys. Rev. Lett. \textbf{122} (16), 160401 (2019).

\bibitem{Teretenkov19}

A.~E. Teretenkov, \textquotedblleft Non-Markovian evolution of multi-level system interacting with several reservoirs. Exact and approximate,\textquotedblright\;Lob. J. Math. \textbf{40} (10), 1587--1605 (2019).

\bibitem{Teretenkov19m}

A.~E. Teretenkov, \textquotedblleft Pseudomode Approach and Vibronic Non-Markovian Phenomena in Light-Harvesting Complexes,\textquotedblright\;Proc. Steklov Inst. Math. \textbf{306}, 242--256 (2019).

\bibitem{Imamoglu94}

A.~Imamoglu, \textquotedblleft Stochastic wave-function approach to non-Markovian systems,\textquotedblright\;Phys. Rev. A \textbf{50} (5), 3650 (1994).

\bibitem{Garraway96}

B.~M. Garraway and P.~L. Knight, \textquotedblleft Cavity modified quantum beats,\textquotedblright\;Phys. Rev. A, \textbf{54} (4), 3592 (1996).

\bibitem{Garraway97}

B.~M. Garraway, \textquotedblleft Nonperturbative decay of an atomic system in a cavity,\textquotedblright\;Phys. Rev. A \textbf{55} (3), 2290 (1997).

\bibitem{Garraway97a}

B.~M. Garraway, \textquotedblleft Decay of an atom coupled strongly to a reservoir,\textquotedblright Phys. Rev. A \textbf{55} (6), 4636 (1997).

\bibitem{Dalton01}

B.~J. Dalton, S~.M. Barnett, and B.~M. Garraway, \textquotedblleft Theory of pseudomodes in quantum optical processes,\textquotedblright\;Phys. Rev. A \textbf{64} (5), 053813 (2001). 

\bibitem{Garraway06}

B.~M. Garraway and B.~J. Dalton, \textquotedblleft Theory of non-Markovian decay of a cascade atom in high-Q cavities and photonic band gap materials,\textquotedblright\;J. of Phys. B \textbf{39} (15), S767 (2006).

\bibitem{Schonleber15}

D.~W. Schonleber, A.~Croy, and A.~Eisfeld, \textquotedblleft Pseudomodes and the corresponding transformation of the temperature-dependent bath correlation function,\textquotedblright\;Phys. Rev. A \textbf{91} (5), 052108 (2015).

\bibitem{Tamascelli18}

D.~Tamascelli, A.~Smirne, S.~F. Huelga, and M.~B. Plenio, \textquotedblleft Nonperturbative treatment of non-Markovian dynamics of open quantum systems,\textquotedblright\;Phys. Rev. Lett. \textbf{120} (3), 030402 (2018).

\bibitem{Mascherpa19}

F.~Mascherpa, A.~Smirne, D.~Tamascelli, P.~F. Acebal, S.~Donadi, S.~F. Huelga, and M.~B. Plenio, \textquotedblleft Optimized auxiliary oscillators for the simulation of general open quantum systems,\textquotedblright\;arXiv:1904.04822 (2019).

\bibitem{Tamascelli19}

D.~Tamascelli, A.~Smirne, J.~Lim, S.~F. Huelga, and M.~B. Plenio, \textquotedblleft Efficient simulation of finite-temperature open quantum systems,\textquotedblright\;Phys. Rev. Lett. \textbf{123} (9), 090402 (2019).

\bibitem{Pleasance20}

G.~Pleasance, B.~M. Garraway, and F.~Petruccione, \textquotedblleft Generalized theory of pseudomodes for exact descriptions of non-Markovian quantum processes,\textquotedblright\;arXiv:2002.09739 (2020).

\bibitem{Friedrichs48}

K.~O. Friedrichs, \textquotedblleft On the perturbation of continuous spectra,\textquotedblright\;Comm. on Pure and Applied Math. \textbf{1} (4), 361--406 (1948).

\bibitem{Teretenkov19OnePart}

A.~E. Teretenkov, \textquotedblleft One-particle approximation as a simple playground for irreversible quantum evolution,\textquotedblright\;arXiv:1912.13123 (2019).

\bibitem{Fleming10}
 
C.~Fleming, N.~I. Cummings, C.~Anastopoulos, and B.~L. Hu,  \textquotedblleft The rotating-wave approximation: consistency and applicability from an open quantum system analysis,\textquotedblright\;J. of Phys. \textbf{43} (40), 405304 (2010).

\bibitem{Tang13}
 
N.~Tang, T.-T.~Xu, and H.-S.~Zeng, \textquotedblleft Comparison between non-Markovian dynamics with and without rotating wave approximation,\textquotedblright\;Chinese Phys. B \textbf{22} (3), 030304 (2013).

\bibitem{Trubilko2020}

A.~I. Trubilko and A.~M. Basharov, \textquotedblleft Theory of relaxation and pumping of quantum oscillator non-resonantly coupled with the other oscillator,\textquotedblright\;Phys. Scr. \textbf{95} (4), 045106 (2020).

\bibitem{Burton05}

T.~A. Burton, \emph{Volterra integral and differential equations} (Elsevier, Amsterdam, 2005).

\bibitem{Teretenkov19exact}

A.~E. Teretenkov, \textquotedblleft Exact non-Markovian evolution with multiple reservoirs,\textquotedblright\;arXiv:1912.13272 (2019).

\bibitem{Breuer02}

H.-P.~Breuer and F.~Petruccione, \emph{The theory of open quantum systems} (Oxford University Press, Oxford, 2002).

\bibitem{Trushechkin19}
A.~Trushechkin, \textquotedblleft Calculation of coherences in Forster and modified Redfield theories of excitation energy transfer,\textquotedblright\;J. of Chem. Phys. \textbf{151} (7), 074101 (2019).

\bibitem{Trushechkin19a}
A.~S. Trushechkin, \textquotedblleft Higher-order corrections to the Redfield equation with respect to the system-bath coupling based on the hierarchical equations of motion,\textquotedblright\;Lob. J. Math. \textbf{40} (10), 1606--1618 (2019).

\bibitem{Agarwal71}

G.~S. Agarwal, \textquotedblleft Brownian motion of a quantum oscillator,\textquotedblright\;Phys. Rev. A \textbf{4} (2) 739 (1971).

\bibitem{Accardi2002}

L.~Accardi, Y.G.~Lu, and I.~Volovich, \emph{Quantum theory and its stochastic limit} (Springer, Berlin, 2002).
	
\end{thebibliography}
\end{document}